\documentclass[submission,copyright,creativecommons]{eptcs}

\usepackage{standalone}
\usepackage[english]{babel}
\usepackage{textcomp}
\usepackage{amsmath, amssymb}
\usepackage{mathtools}
\usepackage{amsthm}
\usepackage{setspace}
\usepackage{hyperref}
\usepackage{csquotes}
\usepackage{booktabs}
\usepackage{array}
\usepackage{ dsfont }
\usepackage{pdfpages}
\usepackage{extarrows}
\usepackage{caption}
\usepackage{subcaption}
\usepackage{verbatim} 
\usepackage{mathtools}
\usepackage{amssymb}
\usepackage{centernot}
\usepackage[h]{esvect}

\newcommand{\act}{Act}
\newcommand{\acttau}{\act_{\tau}}
\newcommand{\wact}{\act^*}

\newcommand{\eps}{\varepsilon}

\newcommand{\trans}[1]{\mathrel{\smash{\xrightarrow{#1}}}} % strong transition with input
\newcommand{\wtrans}[1]{\mathrel{\smash{\xRightarrow{#1}}}} % weak transition with input

\newcommand{\wvt}[1]{\wtrans{\vv{#1}}} % weak vector transition
\newcommand{\wepvt}{\wtrans{\eps}} % weak epsilon vector
\newcommand{\wwvt}{\wvt{w}} % weak w vector
\newcommand{\wvvt}{\wvt{v}} % weak v vectortransition
\newcommand{\wuvt}{\wvt{u}} % weak v vectortransition

\newcommand{\wtwt}{\wtrans{}} % weak hat tau transition
\newcommand{\whtrans}[1]{\wtrans{\hat{#1}}} % weak hat transition
\newcommand{\whal}{\wtrans{\hat{\alpha}}} % weak alpha hat transition

\newcommand{\weakStepDelay}[1]{\mathrel{\xlongequal{#1}\hspace{-3.8pt}\triangleright}}%
\newcommand{\delay}[1]{\weakStepDelay{#1}}

\theoremstyle{definition}
\newtheorem{defi}{Definition}

\newtheorem{example}{Example}

\theoremstyle{plain}
\newtheorem{theorem}{Theorem}
\newtheorem{lemma}[theorem]{Lemma}

\newcommand{\vw}{\vv{w}}
\newcommand{\vu}{\vv{u}}

\newcommand{\vecv}{\vv{v}}

\newcommand{\GM}{\mathcal{G}_{C}}

\newcommand{\gameMove}{\mathrel{\rightarrowtail}}%
\newcommand{\game}{\mathcal{G}}%
\newcommand{\attackerPos}[1]{{(#1)}_\mathtt{a}}
\newcommand{\defenderPos}[1]{{(#1)}_\mathtt{d}}

\newcommand{\CDefSimPos}[1]{\code{Sim}\defenderPos{#1}}
\newcommand{\CDefSwapPos}[1]{\code{Swap}\defenderPos{#1}}

\newcommand{\playsConsF}[2]{G^{\infty}_{#1}[#2]}
\newcommand{\playsConsFInit}[1]{\playsConsF{#1}{\attackerPos{p_0, \set{q_0}}}}

\newcommand{\strat}{f_C}

\newcommand{\pre}[1]{\preceq_{\mathit{#1}}} % preorder
\newcommand{\eq}[1]{\sim_{\mathit{#1}}}   %equiv
\newcommand{\rel}[1]{\mathcal{#1}}

\newcommand{\ccsLaw}[1]{\mathbf{#1}}

\newcommand{\ccsStop}{\boldsymbol{0}}%
\newcommand{\ccsPrefix}{\ldotp\!}%
\newcommand{\ccsChoice}{+}%
\newcommand{\ccsDef}{\mathrel{\stackrel{\mathrm{def}}{=}}}%
\newcommand{\ccsHide}{\mathrel{\backslash}}%
\newcommand{\ccsPar}{\mathrel{\mid}}%
\newcommand{\ccs}{\mathsf{CCS}}%
\newcommand{\ccsIdentifier}[1]{\mathsf{#1}}%
\newcommand{\ccsOutm}[1]{\overline{#1}}%
\newcommand{\action}[1]{\mathsf{#1}}%

\newcommand{\hmlObs}[1]{\langle#1\rangle}

\newcommand{\set}[1]{\{#1\}}
\makeatletter
\providecommand*{\code}[1]{\texttt{#1}}
\makeatother

\newcommand{\isb}{$\vcenter{\hbox{\includegraphics[height=\fontcharht\font`\/]{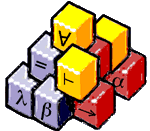}}}$}
\newcommand{\isbref}[3]{\footnote{\isb\,#1 \href{https://concurrency-theory.org/contrasimulation-game/Contrasimulation/\detokenize{#2}.html\##3}{\nolinkurl{#2.#3}}}}

\makeatletter
\DeclareRobustCommand{\rvdots}{%
  \vbox{
    \baselineskip4\p@\lineskiplimit\z@
    \kern-\p@
    \hbox{.}\hbox{.}\hbox{.}
  }}
\makeatother

\makeatletter
\newcommand{\dneg}{\mathord{\vphantom{\neg}\mathpalette\dneg@\relax}}
\newcommand{\dneg@}[2]{%
  \ooalign{%
    $\m@th#1\neg$\cr
    \raisebox{0.8\height}{\clipbox{0pt {0.6\height} 0pt 0pt}{$\m@th#1\neg$}}\cr
  }%
}
\makeatother

\makeatletter
\def\@endtheorem{\endtrivlist}% NEW
\makeatother

\makeatletter
\def\orcidID#1{\href{http://orcid.org/#1}{\protect\raisebox{-1.25pt}{\protect\includegraphics{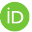}}}}
\makeatother

\usepackage{graphicx}
\usepackage{tikz}
\usetikzlibrary{positioning}
\usetikzlibrary{shapes.geometric, arrows, arrows.meta}
\usetikzlibrary{calc,decorations.pathmorphing,shapes}

\usepackage{trimclip}
\usepackage{ stmaryrd }

\providecommand{\event}{EXPRESS/SOS 2021} % Name of the event you are submitting to

\title{A Game Characterization for Contrasimilarity}

\author{
Benjamin Bisping\orcidID{0000-0002-0637-0171} \qquad Luisa Montanari\orcidID{0000-0002-5270-0290}
\institute{Technische Universtiät Berlin\\ Berlin, Germany}
\email{benjamin.bisping@tu-berlin.de, luisa.montanari@campus.tu-berlin.de}
}

\def\titlerunning{A Game Characterization for Contrasimilarity}
\def\authorrunning{B.\ Bisping \& L.\ Montanari}

\begin{document}

\maketitle

\begin{abstract}
We present the first game characterization of contrasimilarity, the weakest form of bisimilarity. The game is finite for finite-state processes and can thus be used for contrasimulation equivalence checking, of which no tool has been capable to date. A machine-checked Isabelle/HOL formalization backs our work and enables further use of contrasimilarity in verification contexts.
\end{abstract}

\section{Introduction}

Contrasimilarity is the weakest equivalence for systems with internal $\tau$-transitions in Rob van Glabbeek's linear-time--branching-time spectrum~\cite{van1993linear} that instantiates to bisimilarity if there is no internal behavior. It differs from the more commonly used weak bisimilarity in granting the additional equalities\linebreak[3] $\ccsLaw{CS}: \tau \ccsPrefix ( \tau \ccsPrefix X \ccsChoice Y) = \tau \ccsPrefix X \ccsChoice Y$, which it shares with coupled similarity, and $\ccsLaw{C}: a \ccsPrefix (\tau \ccsPrefix X \ccsChoice \tau \ccsPrefix Y) = a \ccsPrefix X \ccsChoice a \ccsPrefix Y$. Contrasimilarity is about ``as weak as one can get'' with respect to $\tau$-steps without losing the distinguishing powers of strong bisimilarity with respect to visible behavior.
%while maintaining the nice property that the notion coincides with strong bisimilarity on systems without internal activity.

Although this position ``on the edge'' makes contrasimilarity (and the contrasimulation preorder) particularly interesting, there has been only little research into it. It has been investigated as an equivalence notion justifying parallelizing transformations in compilers by Bell~\cite{bell_certifiably_2013}, as a strong way of relating context-free processes to their encodings as push-down automata by Baeten et al.~\cite{bctbc2008cfpPda, deHoop2017cfpPdp}, and as the limit of arbitrarily-nested statements about impossible futures by Voorhoeve and Mauw~\cite{voorhoeve_impossible_2001}.

In this paper, we give the first game characterization of the contrasimulation preorder and prove its correctness. All the stronger (branching-time) notions of the linear-time--branching-time spectrum with internal steps already have game characterizations~\cite{escrig_games_2017, bn2019coupledsimTacas}. All the weaker (linear-time) notions can readily be characterized by slightly adapting the games~\cite{ramalingam_game_2008,bn2021ltbtspectroscope} of the linear-time--branching-time spectrum without internal steps~\cite{van2001linear}. So, our contribution inserts the last puzzle piece to have game characterizations for the whole linear-time--branching-time spectrum with internal steps of~\cite{van1993linear}.

\vspace{1em} \noindent
\textbf{Contributions and Structure.\quad} This paper makes the following contributions:

\begin{itemize}
    \item We assemble a concise \emph{overview of the bisimulation-like properties of contrasimilarity} in Section~\ref{sec:contrasim}.
    \item In Section~\ref{sec:contrasim-game}, we present a \emph{new game for the contrasimulation preorder} based on subset constructions, which is finite (albeit exponential) for finite-state processes.
    \item Section~\ref{sec:contrasim-game-correctness} proves the \emph{correctness of the game characterization} by relating defender winning strategies and contrasimulations, which turns out to be slightly trickier than for stronger equivalences.
    \item Section~\ref{sec:related} links our contributions to other research and hints at the relationship to modal-logical characterizations. 
    \item All definitions and proofs of this paper have been \emph{formalized in the interactive proof assistant Isabelle/HOL} \cite{wenzel2004isabelle}. Each lemma comes with a footnote pointing to its Isabelle proof at \url{https://concurrency-theory.org/contrasimulation-game/Contrasimulation/}.
\end{itemize}

\section{Contrasimilarity: The Weakest Bisimilarity \label{sec:contrasim}}

Notions of equivalence formalize when two process models can be considered indistinguishable given a certain model of communication or observation. The most famous such notion is bisimilarity, which holds if two processes can match each other's transitions repeatedly. Contrasimilarity is a reformulation of this, saying that every transition sequence by one process can be matched by the other process in such a way that they could continue with their roles swapping each time. This section explains systems with internal behavior (Subsection~\ref{subsec:systems-internal}) and gives a formal definition of contrasimilarity (Subsection~\ref{subsec:def-contrasim}). We then show that contrasimilarity and bisimilarity actually are the same concept unless there is internal behavior (Subsection~\ref{subsec:contra-bisim}) and briefly look at systems that are not contrasimilar (Subsection~\ref{subsec:contra-not-linear}). For this section, we omit the proofs; the Isabelle proofs can be found via the footnotes.

\subsection{Systems with Internal Steps \label{subsec:systems-internal}}

In order to illustrate what kinds of systems contrasimilarity equates, we start with an example of two systems with internal communication behavior that are equivalent with respect to contrasimilarity, but not with respect to weak bisimilarity. (The example behaves similarly to the one in~\cite{bell_certifiably_2013}. Its transition system will be given in Figure~\ref{fig:c-philosophers}.)

\begin{example}[Dining Hall Philosophers]
\label{exa:philosophers-dining}

Two philosophers $\ccsIdentifier A$ and $\ccsIdentifier B$ want to eat pasta.
They can get their spaghetti ($\action{sp}$) once the dining hall counter opens ($\action{op})$.
However, there is only one plate ($\action{pl}$), which must be taken before picking up the spaghetti (regardless of whether the counter has opened).
We are waiting for them in the dining area outside; so we only observe whether the counter has opened and who eats, whereas the plate and spaghetti grabbing are invisible to us.

\begin{center}
\includegraphics[width=0.5\textwidth]{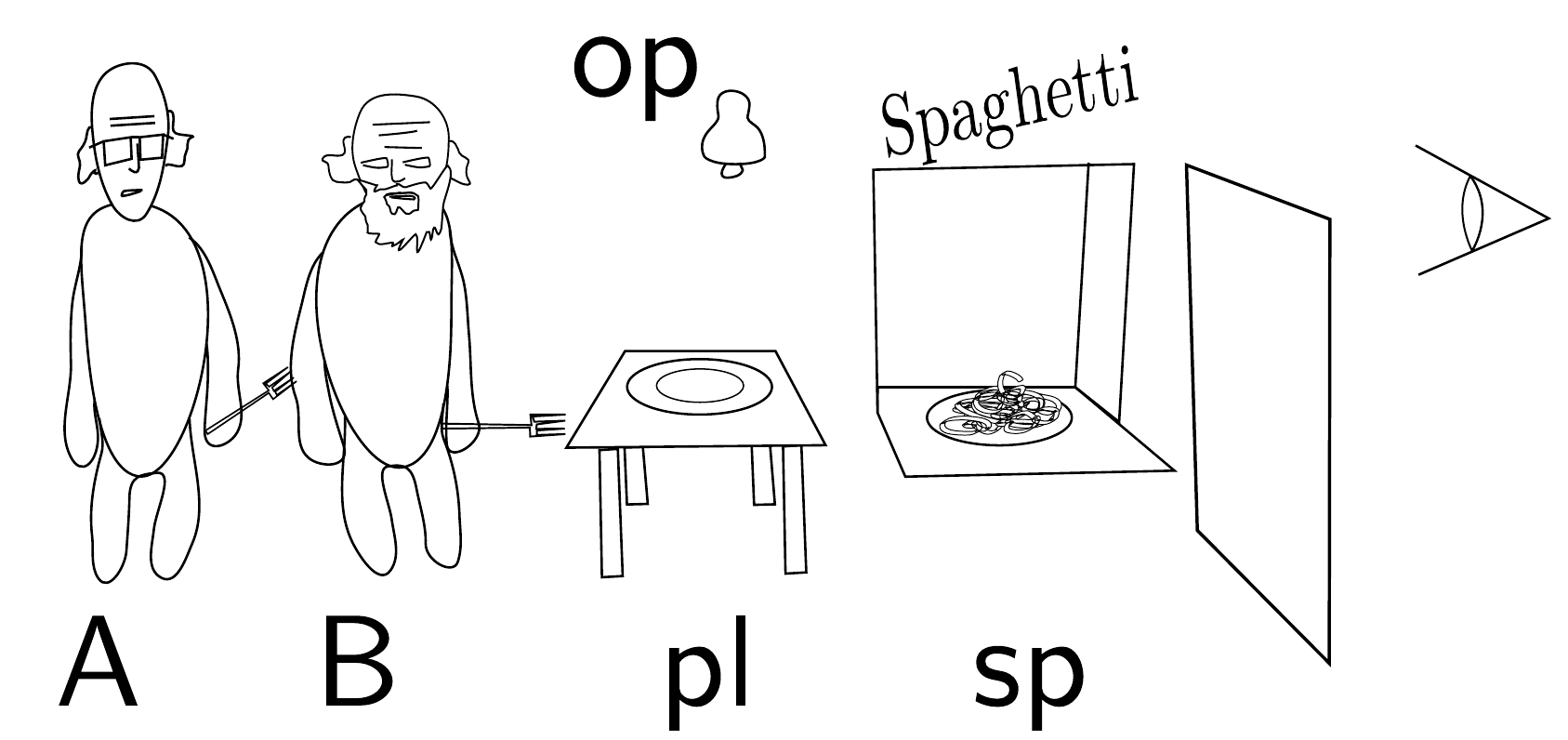}
\end{center}

\noindent
The following $\ccs$ structure models the situation of the philosophers waiting at the counter $\ccsIdentifier{P_c}$ in the notation of Milner~\cite{milner1989communication}.
The resources correspond to sending actions like $\ccsOutm{\action{sp}}$ (which can be consumed only once) and obtaining the resources corresponds to receiving actions like $\action{sp}$. The subprocesses run in parallel ($\ldots \ccsPar \ldots$) and internal communication is hidden from the outside ($\ldots \ccsHide\:\{\action{sp}\}$).
\begin{equation*}
    \ccsIdentifier {P_c} \, \ccsDef \,
    \left(\, \action{pl} \ccsPrefix \action{sp} \ccsPrefix \action{aEats}\,
    \ccsPar \, \action{pl} \ccsPrefix \action{sp} \ccsPrefix \action{bEats}\,
    \ccsPar \, \ccsOutm{\action{pl}} \,
    \ccsPar \, \action{op} \ccsPrefix \ccsOutm{\action{sp}} \,
    \right)\;\ccsHide\:\{\action{pl}, \action{sp}\}\\
\end{equation*}
With $\action{sp}$ being hidden: Does it make a difference to the observer if it is not the spaghetti counter waiting for the opening event before handing out pasta, but rather the philosophers waiting for the event before grabbing the pasta? This would amount to the following model of patient philosophers $\ccsIdentifier{P_p}$.
\begin{equation*}
\ccsIdentifier {P_p} \, \ccsDef \,
    \left(\, 
    \action{pl} \ccsPrefix \action{op} \ccsPrefix \action{sp} \ccsPrefix \action{aEats}\,
    \ccsPar \, \action{pl} \ccsPrefix \action{op} \ccsPrefix \action{sp} \ccsPrefix \action{bEats}\, 
    \ccsPar \, \ccsOutm{\action{pl}} \,
    \ccsPar \, \ccsOutm{\action{sp}} \,
    \right)\;\ccsHide\:\{\action{pl}, \action{sp} \}\\
\end{equation*}
If one's application needs to abstract over “where exactly” internal waiting is happening, and thus equate $\ccsIdentifier {P_c}$ and $\ccsIdentifier {P_p}$, one has to pick contrasimilarity (or a weaker equivalence) for the semantics.
\end{example}

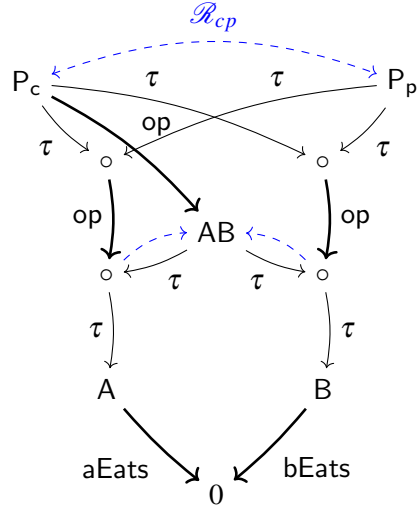
\begin{figure}[t]
  \begin{centering}
\begin{tikzpicture}[->,auto,node distance=2cm, rel/.style={dashed,font=\it, blue}, ext/.style={line width=1pt}]

\node (Pc){$\ccsIdentifier{P_c}$};

\node (Pp) [right of=Pc, node distance=5cm] {$\ccsIdentifier{P_p}$};

\node (A1) [below right of=Pc, node distance=1.5cm] {$\circ$};
\node (A2) [below of=A1, node distance = 1.5cm] {$\circ$};
\node (A3) [below of=A2, node distance = 1.5cm] {$\ccsIdentifier A$};

\node (B1) [below left of=Pp, node distance=1.5cm] {$\circ$};
\node (B2) [below of=B1, node distance = 1.5cm] {$\circ$};
\node (B3) [below of=B2, node distance = 1.5cm] {$\ccsIdentifier B$};

\node (Stop) [below left of=B3, node distance = 2.0cm] {$\ccsStop$};

\node (AB) [above of=Stop, node distance = 3.5cm] {$\ccsIdentifier{AB}$};

\path
(Pc) edge [bend right=10, swap] node {$\tau$} (A1)
(A1) edge [bend left=10, swap, ext] node {$\action{op}$} (A2)
(A2) edge [bend left=10, swap] node {$\tau$} (A3)
(A3) edge [bend right=5, swap, ext] node {$\action{aEats}$} (Stop)
(Pc) edge [bend left=10, pos=0.3] node {$\tau$} (B1)
(B1) edge [bend left=10, ext] node {$\action{op}$} (B2)
(B2) edge [bend left=10] node {$\tau$} (B3)
(B3) edge [bend left=5, ext] node {$\action{bEats}$} (Stop)
(Pc) edge [bend left=10, ext] node {$\action{op}$} (AB)
(AB) edge [bend left=10] node {$\tau$} (A2)
(AB) edge [bend right=10, swap] node {$\tau$} (B2)
;

\path
(Pp) edge[bend right=10, swap, pos=0.3] node {$\tau$} (A1)
(Pp) edge[bend left=10] node {$\tau$} (B1)
;

\path
(Pc) edge [<->, rel, bend left=20] node {$\rel{R}_{cp}$} (Pp)
(A2) edge [rel, bend left=20] node {} (AB)
(B2) edge [rel, bend right=20] node {} (AB)
;

\end{tikzpicture}
\par\end{centering}
\caption{The reflexive closure of $\textcolor{blue}{\rel{R}_{cp}}$ is a contrasimulation on the philosopher system from Example~\ref{exa:philosophers-dining}.
}
\label{fig:c-philosophers}
\end{figure}

% ... 3 pages

%Transition systems and transitions (Def. 1 + Def. 2)
%Simulations etc. (Def. 4 ff.)
%Continue example by giving relations.

\begin{defi}
\label{def lts}
A \emph{labeled transition system}, or \emph{LTS} for short, $(S, \acttau, \xrightarrow{})$ consists of a set of states $S$, of a set $\acttau = \act \cup \{\tau\}$ of visible actions $\act$ and a special \emph{internal action} $\tau \notin \act$, and of a transition relation ${\xrightarrow{}} \subseteq S \times \acttau \times S$.
\end{defi}

\noindent
The transition system for Example~\ref{exa:philosophers-dining} is depicted as the black part in Figure~\ref{fig:c-philosophers}. All deadlocks are combined into the state $\ccsStop$. The transitions where communication internal to the system happens are labeled by $\tau$, which denotes internal behavior. Where there are multiple internal transitions originating from the same process state, the system performs an \emph{internal choice}. In process $\ccsIdentifier{P_p}$, the internal choice happens at the start, whereas in $\ccsIdentifier{P_c}$, the choice is interleaved with the observable occurrence of $\action{op}$.

\begin{defi}
\label{arrow notation} We employ the following notation for transitions in the system (with $p,p',q \in S$; $\alpha \in \acttau$):
\begin{itemize}
  \item Strong steps: $p \trans{\alpha} p'$ iff $(p, \alpha, p') \in {\xrightarrow{}}$.
  \item Internal steps: $p \wtwt p'$ iff $p \trans{\tau}^* p'$.
  \item Delay steps: $p \weakStepDelay{\alpha} p'$ iff $\alpha = \tau$ with $p \wtwt p'$ or $\alpha \in \act$ with $p \wtwt\trans{\alpha} p'$.
  \item Weak steps: $p \whal p'$ iff $p \weakStepDelay{\alpha}\wtwt p'$.
  \item Weak word steps: $p \wvt{w} p'$ iff $p \whtrans{w_0}\whtrans{w_1} ... \whtrans{w_n} p'$ with $\vw = w_0w_1 \ldots w_n \in \wact$ or $p \wtwt p'$ for the empty word $\vw = \eps$.
  \item Absence of steps: $p \centernot{\trans{\alpha}}$, $p \centernot{\whtrans{\alpha}}$ and $p \centernot{\wvt{w}}$ to denote that certain transitions are not possible from $p$. $p$ is called \emph{stable} iff $p \centernot{\trans{\tau}}$.
  \item Lifting of steps to sets: $P \trans{\alpha} P'$ iff $P' = \{p' \mid \exists p \in P : p \trans{\alpha} p'\}$, and $P \trans{\alpha} p'$ iff there is a $p \in P$ with $p \trans{\alpha} p'$; analogously for $\wtwt$, $\weakStepDelay{\alpha}$, $\whtrans{\alpha}$, and $\wvt{w}$.
\end{itemize}
\end{defi}

\noindent
Note that delay steps $\weakStepDelay{a}$ differ from the more common weak steps $\whtrans{a}$ ($a \in \act$) in that the step ends in the strong $\trans{a}$-transition with no trailing internal behavior. Also note that $P \trans{\alpha} \{p'\}$ is stronger than $P \trans{\alpha} p'$ in that it rules out possible other $\trans{\alpha}$-transitions.

\subsection{Defining Contrasimilarity \label{subsec:def-contrasim}}

It is common to define notions of behavioral equivalence and behavioral preorders in terms of relations between process states that fulfill certain properties. For instance, weak similarity is defined in terms of weak simulation relations, which also induce a weak simulation preorder. Contrasimilarity is obtained by making a subtle change to the weak simulation property.

\begin{defi}
\label{weak sim def}
A \emph{weak simulation} is a relation $\rel R$ where, for all $(p,q) \in \rel R$ with $p \trans{\alpha} p'$, there is a $q'$ with $q \whal q'$ and $(p', q') \in \rel R$. We say that $p$ is \emph{weakly simulated by} $q$, written $p \pre{WS} q$, iff there is a weak simulation $\rel R$ with $(p,q) \in \rel R$. If $p \pre{WS} q \pre{WS} p$, the two are weakly similar, written $p \eq{WS} q$. If $\rel R$ is a symmetric weak simulation, the processes are called \emph{weakly bisimilar} ($\eq{WB}$), which implies all of the above.
\end{defi}

\noindent
On the processes of Example~\ref{exa:philosophers-dining}, $\{(p,p) \mid p\in S\} \cup \{ (\ccsIdentifier{P_p}, \ccsIdentifier{P_c})\}$ is a weak simulation, because the steps from $\ccsIdentifier{P_p}$ are a subset of the steps from $\ccsIdentifier{P_c}$. This implies $\ccsIdentifier{P_p} \pre{WS} \ccsIdentifier{P_c}$.

However, there is no weak simulation in the opposite direction, because $\ccsIdentifier {P_c}$ can $\action{op}$-step to a state where $\action{aEats}$ and $\action{bEats}$ are weakly enabled, while $\ccsIdentifier {P_p}$ cannot. Thus, weak similarity and hence also weak bisimilarity distinguish the two systems.

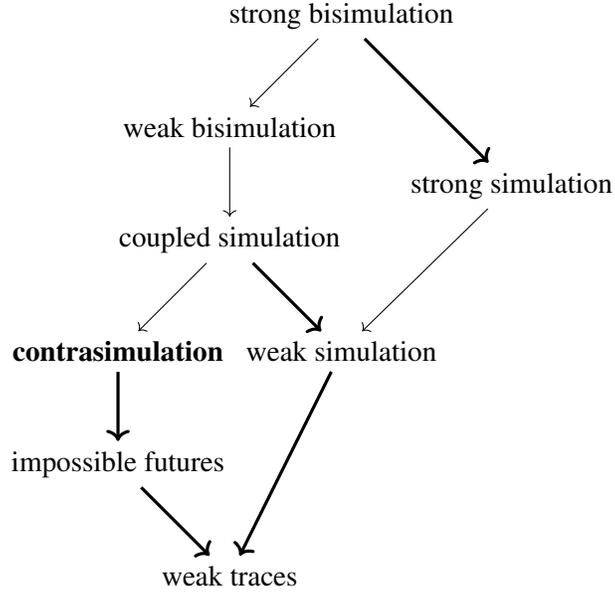
\begin{figure}[t]
  \begin{centering}
\begin{tikzpicture}[->,auto,node distance=1.5cm]

\node (SB){strong bisimulation};
\node (SS) [below right of=SB, node distance=3.2cm] {strong simulation};
\node (WB) [below left of=SB, node distance=2.1cm] {weak bisimulation};
\node (CS) [below of=WB] {coupled simulation};
\node (WS) [below right of=CS, node distance=2.1cm] {weak simulation};
\node (C) [below left of=CS, node distance=2.1cm] {\textbf{contrasimulation}};
\node (IF) [below of=C] {impossible futures};
\node (WT) [below right of=IF, node distance=2.1cm] {weak traces};

\path
(SB) edge (WB)
(SB) edge[line width=1.2pt] (SS)
(SS) edge (WS)
(WB) edge (CS)
(CS) edge (C)
(CS) edge[line width=1.2pt] (WS)
(C) edge[line width=1.2pt] (IF)
(IF) edge[line width=1.2pt] (WT)
(WS) edge[line width=1.2pt] (WT);

\end{tikzpicture}
\par\end{centering}
\caption{Hierarchy of equivalences. Arrows denote implications of preordering. Thinner arrows collapse into bi-implication for systems without internal steps. \label{fig:equivalences}}
\end{figure}

Due to its inherent recursiveness, Definition~\ref{weak sim def} can also be rephrased in terms of words instead of single actions:

\begin{lemma}
$\rel R$ is a weak simulation precisely if, for all $(p,q) \in \rel R$ with $\vw \in \wact$ and $p \wwvt p'$, there is a $q'$ with $q \wwvt q'$ and $(p', q') \in \rel R$.\isbref{lemma}{Weak_Relations}{weak_sim_word}
\end{lemma}

\noindent
This observation motivates why the following definition with $p'$ and $q'$ swapping sides at the end has been named \emph{contra-}simulation:

\begin{defi} 
\label{contra def}
A \emph{contrasimulation} is a relation $\rel R$ where, for all $(p,q) \in \rel R$ with $\vw \in \wact$ and $p \wwvt p'$, there is a $q'$ with $q \wwvt q'$ and $(q', p') \in \rel R$. The contrasimulation preorder $\pre{C}$ and contrasimilarity $\eq{C}$ are defined analogously to Definition~\ref{weak sim def}.
\end{defi}

\noindent
The reflexive closure of the relation in Figure~\ref{fig:c-philosophers}, $\rel{R}_{cp} \cup \{(p,p) \mid p\in S\}$, is a contrasimulation. Hence, $\ccsIdentifier{P_c}$ and $\ccsIdentifier{P_p}$ are contrasimilar, $\ccsIdentifier{P_c} \eq{C} \ccsIdentifier{P_p}$.

Let us quickly stress that the contrasimulation preorder indeed \emph{is} a sensible behavioral preorder (which is not true for its intransitive neighbors eventual simulation~\cite{bell_certifiably_2013} and stability-coupled simulation~\cite{parrow1994complete}):

%\newpage

\begin{lemma}
Properties of $\pre{C}$:

\begin{itemize}
    \item $\pre{C}$ is transitive,\isbref{lemma}{Contrasimulation}{contrasim_trans} as the interleaved concatenation $\rel R_1 \rel R_2 \cup \rel R_2 \rel R_1$ of two contrasimulations $\rel R_1$ and $\rel R_2$ is itself a contrasimulation.\isbref{lemma}{Contrasimulation}{contrasim_trans_constructive}
    \item $\pre{C}$ is reflexive.\isbref{lemma}{Contrasimulation}{contrasim_refl} (More generally: $p \wtwt p'$ implies $p' \pre{C} p$.\isbref{lemma}{Contrasimulation}{contrasim_tau_step})
    \item $\pre{C}$ is itself a contrasimulation\isbref{lemma}{Contrasimulation}{contrasim_preorder_is_contrasim} (and thus the greatest contrasimulation\isbref{lemma}{Contrasimulation}{contrasim_preorder_is_greatest}).
\end{itemize}
\end{lemma}

\subsection{Relationship of Contra- and Bi-similarity \label{subsec:contra-bisim}}

Like the weak simulation preorder, the contrasimulation preorder is not symmetric for most systems. Like with weak simulation, the contrasimulation property can be used to characterize weak bisimulation:

\begin{lemma}[Bisimulation characterization]
If a contrasimulation $\rel R$ is symmetric, then $\rel R$ moreover is a weak simulation. (Meaning that, in this case, $(p,q) \in \rel{R}$ implies $p \eq{WB} q$.)\isbref{lemma}{Contrasimulation}{symm_contrasim_is_weak_bisim}
\end{lemma}

\noindent
Unlike weak simulations, contrasimulations are “symmetric up to internal steps.” We call this property \emph{coupling}, as it differentiates coupled simulations from weak simulations~\cite{bisping2020coupled32}:

\begin{lemma}[Coupling]
If $\rel R$ is a contrasimulation, then $(p, q) \in \rel R$ implies there is a $q'$ such that $q \wtwt q'$ and $(q', p) \in \rel R$.\isbref{lemma}{Contrasimulation}{contrasim_coupled}
\end{lemma}

\noindent
On stable states, coupling equates to local symmetry. This is nice for systems without internal behavior:

\begin{lemma}[Contra/Bi-simulation]\label{lem:contrasim-weakest-bisim}
If $\trans{}$ contains no $\tau$-steps, and $\rel R$ is a contrasimulation, then $\rel R$ is symmetric and thus a bisimulation.\isbref{lemma}{Contrasimulation}{contrasim_weakest_bisim}
\end{lemma}

\noindent
Accordingly, ${\pre{C}} = {\eq{WB}} = {\eq{SB}}$ for systems without internal steps. In this sense, contrasimilarity is closer to bisimilarity than weak similarity: A weak simulation on a $\tau$-free system is just a strong simulation but does not need to be a bi-simulation. The hierarchy is depicted in Figure~\ref{fig:equivalences}.

\subsection{No Shortcuts \label{subsec:contra-not-linear}}

We already mentioned that contrasimilarity grants the equality $a \ccsPrefix (\tau\ccsPrefix X \ccsChoice \tau\ccsPrefix Y) = a \ccsPrefix X \ccsChoice a \ccsPrefix Y$. However, this does not mean that internal choice may commute over actions: The equality $a \ccsPrefix (\tau\ccsPrefix X \ccsChoice \tau\ccsPrefix Y) = \tau\ccsPrefix a\ccsPrefix X \ccsChoice \tau\ccsPrefix a\ccsPrefix Y$ is \emph{not sound} for contrasimilarity, as the following example illustrates:

\begin{example}[Locked-out Philosophers]
\label{exa:philosophers-dining-late}

Consider this slight modification of $\ccsIdentifier{P_c}$, where already the scarce plate $\action{pl}$ is guarded by the opening $\action{op}$:
\begin{equation*}
    \ccsIdentifier {P_l} \, \ccsDef \,
    \left(\, \action{pl} \ccsPrefix \action{sp} \ccsPrefix \action{aEats}\,
    \ccsPar \, \action{pl} \ccsPrefix \action{sp} \ccsPrefix \action{bEats}\, 
    \ccsPar \, \action{op} \ccsPrefix \ccsOutm{\action{pl}} \,
    \ccsPar \, \ccsOutm{\action{sp}} \,
    \right)\;\ccsHide\:\{\action{pl},\action{sp} \}\\
\end{equation*}
As depicted in Figure~\ref{fig:philosophers-dining-late}, the result of this change is that the decision between philosophers $\ccsIdentifier{A}$ and $\ccsIdentifier{B}$ can happen \emph{only} after $\action{op}$. %($\ccsIdentifier{P_l}$ looks like $\ccsIdentifier{P_c}$ in Figure~\ref{fig:c-philosophers} with its two $\tau$-transitions removed.) 

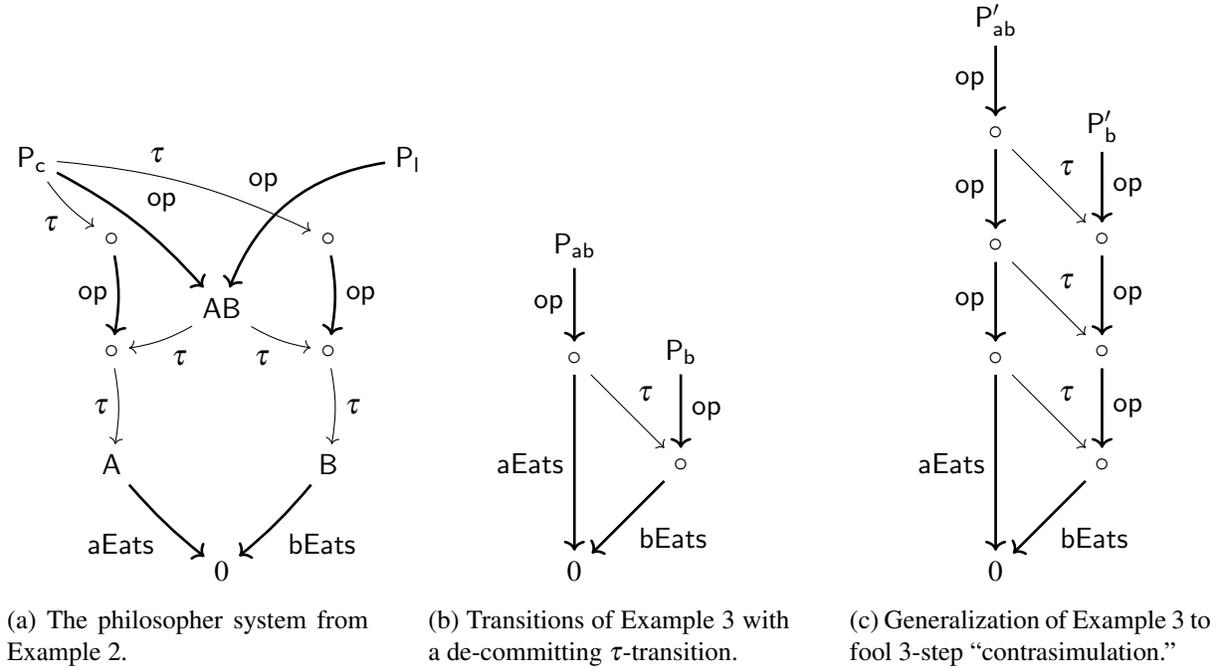
\begin{figure}
     \centering
     \begin{subfigure}[b]{0.3\textwidth}
        \centering
         
        \begin{tikzpicture}[->,auto,node distance=2cm, rel/.style={dashed,font=\it, blue}, ext/.style={line width=1pt}]

            \node (Pc){$\ccsIdentifier{P_c}$};
            
            \node (Pl) [right of=Pc, node distance=5cm] {$\ccsIdentifier{P_l}$};
            
            \node (A1) [below right of=Pc, node distance=1.5cm] {$\circ$};
            \node (A2) [below of=A1, node distance = 1.5cm] {$\circ$};
            \node (A3) [below of=A2, node distance = 1.5cm] {$\ccsIdentifier A$};
            
            \node (B1) [below left of=Pp, node distance=1.5cm] {$\circ$};
            \node (B2) [below of=B1, node distance = 1.5cm] {$\circ$};
            \node (B3) [below of=B2, node distance = 1.5cm] {$\ccsIdentifier B$};
            
            \node (Stop) [below left of=B3, node distance = 2.0cm] {$\ccsStop$};
            
            \node (AB) [above of=Stop, node distance = 3.5cm] {$\ccsIdentifier{AB}$};
            
            \path
            (Pc) edge [bend right=10, swap] node {$\tau$} (A1)
            (A1) edge [bend left=10, swap, ext] node {$\action{op}$} (A2)
            (A2) edge [bend left=10, swap] node {$\tau$} (A3)
            (A3) edge [bend right=5, swap, ext] node {$\action{aEats}$} (Stop)
            (Pc) edge [bend left=10, pos=0.3] node {$\tau$} (B1)
            (B1) edge [bend left=10, ext] node {$\action{op}$} (B2)
            (B2) edge [bend left=10] node {$\tau$} (B3)
            (B3) edge [bend left=5, ext] node {$\action{bEats}$} (Stop)
            (Pc) edge [bend left=10, ext] node {$\action{op}$} (AB)
            (AB) edge [bend left=10] node {$\tau$} (A2)
            (AB) edge [bend right=10, swap] node {$\tau$} (B2)
            ;
            
            \path
            (Pl) edge [bend right=30, swap, ext] node {$\action{op}$} (AB)
            ;
            
        \end{tikzpicture}
        \caption{The philosopher system from Example~\ref{exa:philosophers-dining-late}.}
        \label{fig:philosophers-dining-late}
     \end{subfigure}
     \hfill
     \begin{subfigure}[b]{0.3\textwidth}
         \centering
         
         \begin{tikzpicture}[->,auto,node distance=2cm, rel/.style={dashed,font=\it, blue}, ext/.style={line width=1pt}]

            \node (Pab){$\ccsIdentifier{P_{ab}}$};
            \node (Pb) [below right of=Pab, node distance=2cm] {$\ccsIdentifier{P_b}$};
            
            \node (A1) [below of=Pab, node distance=1.5cm] {$\circ$};
            
            \node (B1) [below of=Pb, node distance=1.5cm] {$\circ$};
            
            \node (Stop) [below left of=B1, node distance = 2.0cm] {$\ccsStop$};
            
            \path
            (Pab) edge[swap, ext] node {$\action{op}$} (A1)
            (A1) edge[swap, ext] node {$\action{aEats}$} (Stop)
            
            (Pb) edge[ext] node {$\action{op}$} (B1)
            (B1) edge[ext] node {$\action{bEats}$} (Stop)
            
            (A1) edge node {$\tau$} (B1);
        
        \end{tikzpicture}
         
        \caption{Transitions of Example~\ref{exa:instable-choice} with a de-committing $\tau$-transition.}
        \label{fig:de-committing-1}
     \end{subfigure}
     \hfill
     \begin{subfigure}[b]{0.3\textwidth}
         \centering
         
         \begin{tikzpicture}[->,auto,node distance=2cm, rel/.style={dashed,font=\it, blue}, ext/.style={line width=1pt}]

            \node (Pab){$\ccsIdentifier{P_{ab}'}$};
            \node (Pb) [below right of=Pab, node distance=2cm] {$\ccsIdentifier{P_b'}$};
            
            \node (A1) [below of=Pab, node distance=1.5cm] {$\circ$};
            \node (A2) [below of=A1, node distance = 1.5cm] {$\circ$};
            \node (A3) [below of=A2, node distance = 1.5cm] {$\circ$};
            
            \node (B1) [below of=Pb, node distance=1.5cm] {$\circ$};
            \node (B2) [below of=B1, node distance = 1.5cm] {$\circ$};
            \node (B3) [below of=B2, node distance = 1.5cm] {$\circ$};
            
            \node (Stop) [below left of=B3, node distance = 2.0cm] {$\ccsStop$};
            
            \path
            (Pab) edge[swap, ext] node {$\action{op}$} (A1)
            (A1) edge[swap, ext] node {$\action{op}$} (A2)
            (A2) edge[swap, ext] node {$\action{op}$} (A3)
            (A3) edge[swap, ext] node {$\action{aEats}$} (Stop)
            
            (Pb) edge[ext] node {$\action{op}$} (B1)
            (B1) edge[ext] node {$\action{op}$} (B2)
            (B2) edge[ext] node {$\action{op}$} (B3)
            (B3) edge[ext] node {$\action{bEats}$} (Stop)
            
            (A1) edge node {$\tau$} (B1)
            (A2) edge node {$\tau$} (B2)
            (A3) edge node {$\tau$} (B3);
        
        \end{tikzpicture}
         
        \caption{Generalization of Example~\ref{exa:instable-choice} to fool 3-step “contrasimulation.”}
        \label{fig:generalizing-de-committing}
     \end{subfigure}
     
        \caption{Processes that are not contrasimilar.}
        \label{fig:counterexamples}
\end{figure}

The change is noticed by contrasimilarity, i.e.\ $\ccsIdentifier{P_c} \not\pre{C} \ccsIdentifier{P_l}$. The reason is that $\ccsIdentifier{P_c} \pre{C} \ccsIdentifier{P_l}$ with the left process resolving its choice would imply $\ccsIdentifier{P_l} \pre{C} \action{op} \ccsPrefix \tau \ccsPrefix \action{aEats}$. But this does not hold because $\ccsIdentifier{P_l} \wtrans{\action{op},\action{bEats}} \ccsStop$ but not $\action{op} \ccsPrefix \tau \ccsPrefix \action{aEats} \wtrans{\action{op},\action{bEats}} \ccsStop$.
\end{example}

\noindent
It is also worthwhile to observe that, in general, the definition of contrasimulation (Definition \ref{contra def}) cannot straight-forwardly be simplified to use single steps $\whtrans{\alpha}$ instead of words $\wvt{a}$, as is the case with its finer siblings like weak bisimulation. Such a simplification is used by de Hoop~\cite{deHoop2017cfpPdp}. However, this is not sound for general systems due to the alternating sides in the definition.

\begin{example}[Instable choice]\label{exa:instable-choice}

Consider $\ccsIdentifier{P_{ab}} \ccsDef \action{op}\ccsPrefix(\action{aEats} \ccsChoice \tau\ccsPrefix\action{bEats})$ and $\ccsIdentifier{P_b} \ccsDef \action{op}\ccsPrefix\action{bEats}$ whose transitions can be found in Figure~\ref{fig:de-committing-1}. The processes are clearly not even weakly trace equivalent. But single-step “contrasimulation” cannot tell them apart, as neither $\ccsIdentifier{P_{ab}} \whtrans{\action{op}} \action{aEats} \ccsChoice \tau\ccsPrefix\action{bEats}$ nor $\ccsIdentifier{P_{ab}} \whtrans{\action{op}} \action{bEats}$ (matched by $\ccsIdentifier{P_{b}} \whtrans{\action{op}} \action{bEats}$) lead to a situation that would not (contra-)simulate $\action{bEats}$.

The counter-example can be generalized as hinted at in Figure~\ref{fig:generalizing-de-committing}, which would need at least a word length of four to distinguish the named processes.
\end{example}

\section{The Contrasimulation Game \label{sec:contrasim-game}}
% ... 3 pages

The contrasimulation preorder can be characterized by a game between two opposing players, the \emph{attacker} and the \emph{defender}. For a given $p, q \in S$, the attacker seeks to disprove $p \pre{C} q$, while the defender seeks to maintain $p \pre{C} q$. We first introduce some general thoughts about such games in Subsection~\ref{subsec:games-prelims}, and then present the contrasimulation game at the core of our contribution in Subsection~\ref{subsec:the-contrasim-game}.

\subsection{Preliminaries \label{subsec:games-prelims}}

For this paper, we use Gale-Stewart-style games in the tradition of Stirling~\cite{stirling1999bisimulation} where the attacker wins by getting the defender stuck, and the defender wins by not getting stuck.

%\newpage %else the definition gets split in the middle
%
\begin{defi}[Games]
  \label{def:simple-game}
  A \emph{simple reachability game} $\game[g_{0}] = (G, G_d, \gameMove, g_0)$ consists of
  \begin{itemize}
    \item a set of \emph{game positions} $G$, partitioned into
    \begin{itemize}
      \item a set of \emph{defender positions} $G_{d}\subseteq G$
      \item and \emph{attacker positions} $G_{a} := G \setminus G_{d}$,
    \end{itemize}
    \item a graph of \emph{game moves} ${\gameMove}\subseteq G\times G$, and
    \item an \emph{initial position} $g_{0}\in G$.
  \end{itemize}
\end{defi}

\begin{defi}[Plays and wins]
  \label{def:plays-wins}
  We call the infinite and finite paths $g_{0}g_{1}\ldots\in G^{\infty}$ with $g_{i}\gameMove g_{i+1}$ \emph{plays} of $\game[g_{0}]$. The defender \emph{wins} infinite plays. If a finite play $g_{0}\dots g_{n}\!\centernot{\gameMove}$ is stuck, the stuck player loses: The defender wins if $g_{n}\in G_{a}$, and the attacker wins if $g_{n}\in G_{d}$. Equivalently, the defender wins precisely those plays in which the defender is not stuck. 
\end{defi}
\begin{defi}[Strategies and winning strategies]
  \label{def:strategies}
  A \emph{defender strategy} is a partial mapping from initial play fragments to next moves $f \subseteq \set{(g_0 \ldots g_n, g_{n+1}) \mid g_n \in G_d \land g_n \gameMove g_{n+1}}$. 
  %A strategy $f$ is called positional if $f(g_0 \ldots g_n)$ only depends on the last position $g_n$.
  A play $g$ is consistent with a defender strategy $f$ iff, for each move $g_i \gameMove g_{i+1}$ with $g_i \in G_d$ where $f(g_0 \ldots g_i)$ is defined, we have $g_{i+1} = f(g_0 \ldots g_i)$. We denote the set of plays $g$ consistent with $f$ by $\playsConsF{f}{g_0}$ for the initial game position $g_0$. If every stuck or infinite play $g \in \playsConsF{f}{g_0}$ is won by the defender, then $f$ is a winning strategy for the defender. The player with a winning strategy for $\game[g_{0}]$ is said to \emph{win} $\game[g_{0}]$.
\end{defi}

% \begin{defi}[Winning regions]
%   \label{def:winning-region}
%   The set $W_a \subseteq G$ of all positions $g$ where the attacker wins $\game[g]$ is called the \emph{attacker winning region} (defender winning region $W_d$ analogous).
% \end{defi}
%
%

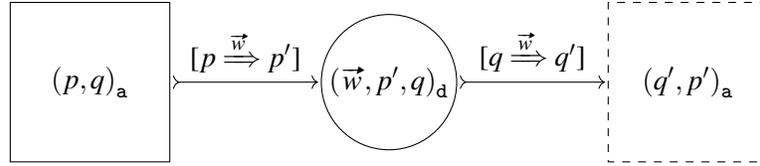
\begin{figure}[t]
    \centering
    
    \begin{tikzpicture}%
      [%>=stealth,
       %shorten >=1pt,
       node distance=2cm,
       shorten <=1pt,
       shorten >=1pt,
       auto,
       baseline={([yshift={-1ex}]current bounding box.north)}
      ]
    
    \tikzset{%
        square/.style={regular polygon,regular polygon sides=4, minimum size = 3cm, inner sep=0cm}
        }
        \tikzstyle{mycircle}=[circle, minimum size = 1.8cm, draw=black, inner sep=0cm]
    
        \node (0) [square, draw=black]                             {$\attackerPos{p,q}$};
        \node (2) [mycircle, draw=black, right=of 0]          {$\defenderPos{\vw, p', q}$};
        \node (1) [square, draw=black, dashed, right=of 2]             {$\attackerPos{q',p'}$};
    
    \path[>->]
    %   FROM        BEND/LOOP  POSITION OF LABEL            LABEL               TO
        (0)         edge   node    {[$ p \wwvt p'$]}    (2)
        (2)         edge    node    {[$ q \wwvt q'$]}    (1)
       ;
    \path[-]
        %(0)         edge [bend left = 20, dashed]                   node    {}   (1)
    ;
    \end{tikzpicture}
    
    \caption{Schematic basic contrasimulation game. Boxes denote attacker positions, circles denote
defender positions and arrows denote game moves. Each game move is only possible if the condition in square brackets is satisfied. Dashed boxes are attacker positions with a new variable assignment and admit analogous moves to the solid boxes.}
    \label{fig:basic_game}
\end{figure}

\noindent
All simple games are \emph{determined}, that is, either the defender or the attacker wins. The winning regions of finite simple games can be computed in linear time in the number of game moves (cf.\ \cite{graedel2007finite}). Therefore, it is desirable to find finite game characterizations of equivalences. For many weak equivalences, such characterizations can be obtained directly from their standard coinductive characterization.

For contrasimilarity, this direct route is less helpful due to its word-based definition. That is why we only briefly discuss it here, and then move on to a different approach in the next subsection.

If we directly transfer Definition~\ref{contra def} into a game, we arrive at the following game alternating between attacker and defender: The attacker may challenge $p \pre{C} q$ by selecting a word $\vw \in \wact$ and a $p' \in S$ with $p \wwvt p'$, to which the defender has to name a $q' \in S$ with $q \wwvt q'$. The sides of the game then swap, and the attacker can go on to question $q' \pre{C} p'$. The attacker wins if the defender is unable to answer with an appropriate $q'$, and the defender wins if the play goes on forever. %, then $p \pre{C} q$ was disproved and the attacker wins. If the attacker runs out of possible challenges or the play goes on forever, then the defender was able to maintain $p \pre{C} q$ and wins. 
%
%
%and challenge the defender to name a $q' \in S$ with $q \wwvt q'$, so that the attacker may again challenge $q' \pre{C} p'$ (with the states $p', q'$ flipped). 

A schematic model of the basic contrasimulation game is given in Figure~\ref{fig:basic_game}. We proved its correctness in the Isabelle formalization.\isbref{theorem}{Basic_Contrasim_Game}{winning_strategy_in_basic_game_iff_contrasim}
Unfortunately, in finite-state processes with cycles, the attacker can challenge with infinitely many words. So, such games will have an infinite number of game positions, even for finite systems. Therefore, we decided not to devote more space to this game in this paper and rather move on to a different game in the following subsection.

\subsection{The Game \label{subsec:the-contrasim-game}}

Let us now turn to the contrasimulation game that is the core contribution of this paper. The idea is to restrict the players' moves to single actions $\alpha \in \acttau$. This means breaking the attacker's word challenge for $p \pre{C} q$ into a simulation phase and a swap request. During the simulation phase, the defender plays a \emph{set of states}:

\begin{figure}[t]
    \centering
    
    \begin{tikzpicture}%
      [%>=stealth,
       shorten <=1pt,
       shorten >=1pt,
       node distance=2cm,
       auto,
       baseline={([yshift={-1ex}]current bounding box.north)}
      ]
    
    \tikzset{%
        square/.style={regular polygon,regular polygon sides=4, minimum size = 3.1cm, inner sep=0cm}
        }
        \tikzstyle{mycircle}=[circle, minimum size = 2.3cm, draw=black, inner sep=0cm]
    
        \node (Ta0) [square, draw]                                      {$\attackerPos{p, Q}$};
        \node (Td0) [mycircle, draw, above right = -5mm and 1.6cm of Ta0]  {$\CDefSimPos{a, p', Q}$};
        \node (Td1) [mycircle, draw, below right= -5mm and 1.6cm of Ta0]   {$\CDefSwapPos{p', Q}$};
        \node (Ta1) [square, draw, dashed, right= of Td0]               {$\attackerPos{p', Q'}$};
        \node (Ta2) [square, draw, dashed, right= of Td1]               {$\attackerPos{q', \{p'\}}$};
    
    \path[>->]
    %   FROM        BEND/LOOP       POSITION OF LABEL   LABEL               TO
        (Ta0)       edge[bend left = 27]        node    {[$p \delay{a} p'$]}    (Td0)
        (Ta0)       edge[swap,bend right = 27]  node    {[$p \wtwt p'$]}        (Td1)
        
        (Td0)       edge                        node    {[$Q \delay{a} Q'$]}    (Ta1)
        (Td1)       edge[swap]                  node    {[$Q \wtwt q'$]}        (Ta2)
       ;
    \path[-, dashed]
        %(Ta1)         edge[bend left = 14]   node    {}   (Ta0)
        %(Ta2)         edge[bend right = 14]  node    {}   (Ta0)
    ;
    
    \end{tikzpicture}
    
    \caption{Schematic contrasimulation set game (Definition~\ref{def:contrasim-set-game}).}
    \label{fig:set_game}
\end{figure}
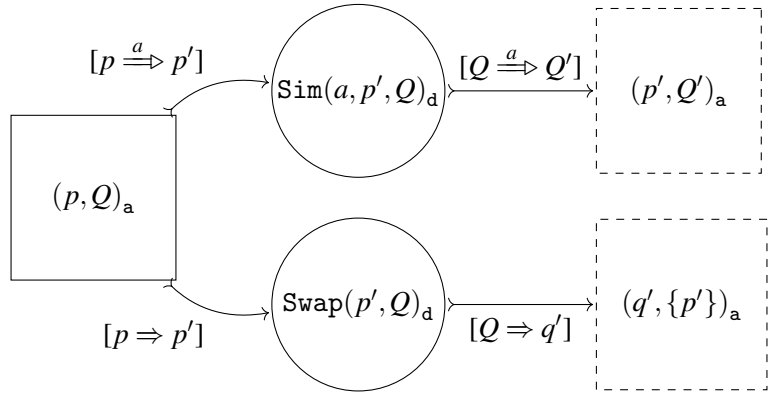

\begin{defi}[$\pre{C}$ game] \label{def:contrasim-set-game}
For a transition system $(S, \acttau, \trans{})$, the \emph{contrasimulation set game} $\GM[g_0] = (G, G_d, \gameMove, g_0)$ consists of 
%\begin{itemize}
%    \item attacker positions \quad $\attackerPos{p, Q} \in G_a$ \quad with $p \in S, Q \subseteq S$,
%    \item defender simulation positions \quad $\CDefSimPos{a, p, Q} \in G_d$ \quad with $a \in Act, p \in S, Q \subseteq S$, and 
%    \item defender swapping positions \quad $\CDefSwapPos{p, Q} \in G_d$ \quad with $p \in S, Q \subseteq S$
%\end{itemize}
%
\begin{itemize}
    \item \makebox[5cm][l]{attacker positions} \makebox[35mm][l]{$\attackerPos{p, Q} \in G_a$} with $p \in S, Q \subseteq S$,
    \item \makebox[5cm][l]{defender simulation positions} \makebox[35mm][l]{$\CDefSimPos{a, p, Q} \in G_d$} with $a \in Act, p \in S, Q \subseteq S$, and 
    \item \makebox[5cm][l]{defender swapping positions} \makebox[35mm][l]{$\CDefSwapPos{p, Q} \in G_d$} with $p \in S, Q \subseteq S$
\end{itemize}

and the following game moves

\begin{itemize}
    \item \makebox[4cm][l]{simulation challenges} \makebox[45mm][l]{$\attackerPos{p, Q} \gameMove \CDefSimPos{a, p', Q}$} if $p \delay{a} p'$,
    \item \makebox[4cm][l]{swap challenges} \makebox[45mm][l]{$\attackerPos{p, Q} \gameMove \CDefSwapPos{p', Q}$} if $p \wtwt p'$,
    \item \makebox[4cm][l]{simulation answers} \makebox[45mm][l]{$\CDefSimPos{a, p', Q} \gameMove \attackerPos{p', Q'}$} if $Q \delay{a} Q'$, and
    \item \makebox[4cm][l]{swap answers} \makebox[45mm][l]{$\CDefSwapPos{p', Q} \gameMove \attackerPos{q', \{p'\}}$} if $Q \wtwt q'$.\isbref{locale}{Contrasim_Set_Game}{c_set_game}
\end{itemize}
\end{defi} 

\noindent
To check whether $p \pre{C} q$ holds we play the contrasimulation set game from the initial attacker position $\attackerPos{p, \set{q}}$.
A schematic model of the game is given in Figure~\ref{fig:set_game}. 

In each simulation phase, the attacker challenges the defender to successively simulate the actions $w_i$ of a word $\vw = w_0 \ldots w_n$. Here, the defender does not play a single state, but rather the \emph{set} of \emph{all} states reachable by the known prefix $w_0 \ldots w_k$ of $\vw$. When the attacker is done choosing $\vw$ in $p \wwvt p'$, they request a swap. The defender must then select a \emph{specific} state $q'$ with $q \wwvt q'$ from their state set.
The players then change sides and the attacker may now challenge $q' \pre{C} p'$.
Hence, the defender postpones the decision of how exactly to simulate the challenged word until a swap is requested.

%In the simulation phase, the attacker may challenge the defender to successively simulate every action $w_i$ in a word $\vw = w_0 \ldots w_n$. The sequence is finished by a swap request, in which the attacker and the defender change sides. The defender, however, postpones the decision how to simulate the word up to the swap request. The defender does not play a single state, but rather a \emph{set of states} $Q$. $Q$ contains all states possibly able to simulate the known prefix $w_0 \ldots w_k$ of $\vw$. Only when the attacker is done choosing $\vw$ in $p \wwvt p'$, signaled by a swap request, does the defender select a specific state $q'$ with $q \wwvt q'$ from this state set. After changing sides, the attacker may now challenge $q' \pre{C} p'$.

\begin{example}[Contrasimulation game on $\ccsIdentifier{P_c}$, $\ccsIdentifier{P_p}$]
A possible play of $\GM$ for the philosopher transition system from Example~\ref{exa:philosophers-dining} would be:
$\attackerPos{\ccsIdentifier{P_c}, \{\ccsIdentifier{P_p}\}}
\gameMove \CDefSimPos{\action{op}, \ccsIdentifier{AB}, \{\ccsIdentifier{P_p}\}}
\gameMove \attackerPos{\ccsIdentifier{AB}, \{ \tau \ccsPrefix \action{aEats}, \tau \ccsPrefix \action{bEats}\}}
\gameMove \CDefSwapPos{\ccsIdentifier{\tau \ccsPrefix \action{aEats}}, \{ \tau \ccsPrefix \action{aEats}, \tau \ccsPrefix \action{bEats}\}}
{\color{blue}\gameMove} \attackerPos{\ccsIdentifier{\tau \ccsPrefix \action{aEats}}, \{ \tau \ccsPrefix \action{aEats}\}}
\gameMove \CDefSimPos{\action{aEats}, \ccsStop, \{ \tau \ccsPrefix \action{aEats}\}}
{\gameMove} \attackerPos{\ccsStop, \{ \ccsStop \}}
\break \gameMove \CDefSwapPos{\ccsStop, \{ \ccsStop \}} \gameMove \dots$ .
The play ends in an infinite loop of swaps and is thus won by the defender. The crucial point of this game is the defender move highlighted in blue, where the defender answers a swap challenge by matching the processes on both sides. After this, it becomes impossible for the attacker to win. If the defender picked $\tau \ccsPrefix \action{bEats}$ instead, they would lose. However, the defender has a winning strategy no matter which moves the attacker chooses for $\GM[\attackerPos{\ccsIdentifier{P_c}, \{\ccsIdentifier{P_p}\}}]$.
\end{example}

\noindent
%The subset construction for the defender states helps to address several problems:
There are no easy ways of switching to single-action moves without using the subset construction:
The defender does not know the full word $\vw$ that the attacker will choose, but only the word prefix $w_0 \ldots w_k$ challenged thus far up to a point $k \leq n$.
Deciding for single states early would thus put the defender at a disadvantage: There might be several states $q'$ with $q \wvt{w_0 \ldots w_k} q'$ for every such $k$, of which only some also satisfy $q \wvt{w_0 \ldots w_k} q' \wvt{w_{k+1} \ldots w_n} q''$ for any $q'' \in S$.
Dually, forcing early swapping would be disadvantageous to the attacker when the attack has to pass through instable states as seen in Example~\ref{exa:instable-choice}.

Crucially, this construction yields a finite game for any finite-state process. As with the well-known subset construction when transforming nondeterministic into deterministic finite automata, the game size is exponential in the size of the state space $S$.

We chose to present the game as an alternating game where attacker and defender take turns. Several modifications could be made to simplify some aspects of the game: 
Note, for example, that the $\CDefSimPos{\ldots}$-positions are not strictly necessary, as the defender has exactly one move originating from each such position. Additionally, parts of the game moves could be broken up into smaller steps on $\trans{}$ instead of $\weakStepDelay{}$. However, both changes would make the game non-alternating. For the purpose of this paper, especially for intuitive proofs like in the following section, we consider the alternating formulation superior.

% \begin{comment}
% \begin{figure}[t]
%     \centering
%     \input{tikz_files/Set_Game_visualized.tikz}
%     \caption{The Set Game visualized}
%     \label{fig:set_game_replication}
% \end{figure}
% \end{comment}

%\subsection{Games Semantics \label{subsec:games}}

%Game definition (Def. 7 ff.) (just referencing the possibility of winning region computation)
%(Trivial) Contrasim word game (if we have space, otherwise just hint to it at the end)
%Contrasim Set Game (Def. 15 + Figures!)
%Continue example by giving game for example system

%Let us fix some notions for \emph{Gale-Stewart-style simple games} where the defender wins all infinite plays.

%where it is easy to compute which player can win from a given game position.

%This is why the spectroscopy game of this paper can easily be used in algorithms. It derives from the following game.%All of them are in some way derived from the following interpretation game for $\hml$.

%---------------------------

\section{Correctness of the Contrasimulation Game \label{sec:contrasim-game-correctness}}

Let us now demonstrate that the $\pre{C}$ game does indeed correspond to the contrasimulation preorder in the sense that the defender wins the game $\GM[\attackerPos{p, \set{q}}]$ precisely if $p \pre{C} q$. In Subsection~\ref{subsec:contraim-game-sound}, we first establish soundness of the characterization, that is, defender winning strategies in the game imply contrasimulations on the LTS. Subsection~\ref{subsec:contraim-game-complete} then shows completeness of the characterization by constructing a defender winning strategy from the greatest contrasimulation.

Our proofs must bridge the gap between the single-action game and the word-transition definition of the contrasimulation property. While transition relations on single actions usually are non-deterministic, the transitions lifted to sets of states \emph{are} deterministic. In order to exploit this in proofs, we first define the word successor function from delay steps:

\begin{defi}
\label{succ func} 
%Let $q, q' \in S$ be states and let $Q \subseteq S$ be a set of states. Also let $a \in Act$ be an action and $\vw \in \wact$ be an action word.
We define the word successor function $succs : \wact \times 2^S \rightarrow 2^S$ recursively as follows:\isbref{primrec}{Weak_Transition_Systems}{dsuccs_seq_rec}
\begin{alignat*}{2}
&succs(\eps, Q) &&= Q \\
&succs(\vw a, Q) &&= \{q' \mid succs(\vw, Q) \delay{a} q'\}
%&succs(\vw a, Q) &&= \{q' \mid \exists q \in succs(\vw, Q) : q \delay{a} q'\}
\end{alignat*}
\end{defi}

\noindent
Intuitively, $succs$ computes the set of states reachable with a given word $\vw$ from a starting set $Q$. Thus, we can use $succs$ to compute the state set of defender simulation positions in the contrasimulation set game.
Note that $succs(\vw, Q)$ will return the empty set if no state in $Q$ admits a $\wwvt$-transition.
\begin{lemma}
\label{in succs iff word-reachable}
Let $\vw \in \wact$ be a word and let $q, q'$ be states in $S$. Then $q \wwvt q'$ implies $succs(\vw, \{q\}) \wtwt q'$.\isbref{lemma}{Weak_Transition_Systems}{word_reachable_implies_in_dsuccs} Furthermore, $q' \in succs(\vw, \{q\})$ implies $q \wwvt q'$.\isbref{lemma}{Weak_Transition_Systems}{in_dsuccs_implies_word_reachable}
\end{lemma}

\subsection{Soundness of the Contrasimulation Game \label{subsec:contraim-game-sound}}

Let us first prove that the $\pre{C}$ game is sound with respect to the contrasimulation preorder, that is, $p \pre{C} q$ holds if the defender wins $\GM[\attackerPos{p, \set{q}}]$. To this end, we first prove an intermediate result stating that the defender is always able to answer simulation challenges over the prefix $\vecv$ of a nonempty word $\vw = \vecv w_n$ in $p \wwvt p'$.

\begin{lemma}[Word challenge building]
\label{succs in play}
Let $f$ be a defender strategy on $\GM[g_0]$ for some $g_0 \in G$ and let $\attackerPos{p, \{q\}}$ be an attacker position in a play consistent with $f$. Let $\vw = \vecv w_n \in \wact$ be a nonempty word and assume $p \wwvt p'$ for some $p' \in S$. Then there exist $p_0, p_1 \in S$ with $p \wvvt p_0 \delay{w_n} p_1 \wtwt p'$ such that the defender position $\normalfont{\CDefSimPos{w_n, p_1,\, succs(\vecv, \{q\})}}$ can be reached in some play consistent with $f$.%
\isbref{lemma}{Contrasim_Set_Game}{def_sim_pos_with_prefix_in_play}
\end{lemma}
\begin{proof}
We will prove this by nonempty induction on $\vw = \vecv w_n \in \wact$: 
\begin{itemize}
    \item \textit{Base Case:} $\vw = w_n = \eps w_n$.
    
    From $p \wtrans{w_n} p'$ we know there exists a $p_1$ such that $p \delay{w_n} p_1 \wtwt p'$. Hence, the attacker can move from $\attackerPos{p, \{q\}}$ to $\CDefSimPos{w_n, p_1, \{q\}} = \CDefSimPos{w_n, p_1,\, succs(\eps, \{q\})}$. This play is still consistent with $f$ as $f$ is only defined for defender moves. 
    
    \item \textit{Induction step:} $\vw = \vecv w_n$ for some nonempty word $\vecv \in \wact$.
    
    Then there exists a $\vu \in \wact$ such that $\vecv = \vu w_{n-1}$ and $p \wuvt p_0 \wtrans{w_{n-1}} p_1 \wtrans{w_n} p'$ for some $p_0, p_1 \in S$. We assume the lemma holds for $\vecv$, i.e. there exists a $p_{01} \in S$ with $p_0 \delay{w_{n-1}} p_{01} \wtwt p_1$ such that the position $\CDefSimPos{w_{n-1}, p_{01},\, succs(\vu, \{q\})}$ can be reached in some play consistent with $f$.
    
    Because there always exists a (potentially empty) set $Q'$ with $succs(\vu, \{q\}) \delay{w_{n-1}} Q'$, the defender is not stuck at $\CDefSimPos{w_{n-1}, p_{01},\, succs(\vu, \{q\})}$. There must therefore exist a position $\attackerPos{p_{01}, Q'} = f(g_0 \ldots \CDefSimPos{w_{n-1}, p_{01},\, succs(\vu, \{q\})}$ that the defender can move to. For $Q'$ we have
    \[{Q' = \{q' \mid succs(\vu, \{q\}) \delay{w_{n-1}} q'\}} = succs(\vu w_{n-1}, \{q\}) = succs(\vecv, \{q\}).\]

From our assumptions $p_{01} \wtwt p_1$ and $p_1 \wtrans{w_n} p'$ we can infer the existence of a $p_1'$ with $p_{01} \wtwt p_1 \delay{w_n} p_1' \wtwt p'$, which we can shorten to $p_{01} \delay{w_n} p_1' \wtwt p'$. Hence, the attacker can move from $\attackerPos{p_{01}, Q'}$ to $\CDefSimPos{w_{n}, p_1', Q'} = \CDefSimPos{w_{n}, p_1',\, succs(\vecv, \{q\})}$, and therefore the defender position $\CDefSimPos{w_{n}, p_1',\, succs(\vecv, \{q\})}$ is reachable in a play consistent with $f$. Furthermore, the second implication $p \wvvt p_{01} \delay{w_n} p_1' \wtwt p'$ follows immediately from $p \wuvt p_0 \delay{w_{n-1}} p_{01} \delay{w_n} p_1' \wtwt p'$ and $\vecv = \vu w_{n-1}$.\qedhere

\end{itemize}
%\vspace{-\topsep}
\end{proof}

\noindent
With this result, we are now able to prove the soundness of the $\pre{C}$ game. 

\begin{lemma}[Soundness]
\label{soundness_menge}
Let $f$ be a winning strategy for the defender on $\GM[\attackerPos{p_0, \set{q_0}}]$ for some $p_0, q_0 \in S$. Then we have $p_0 \pre{C} q_0$.
\isbref{lemma}{Contrasim_Set_Game}{set_contrasim_game_sound}
\end{lemma}
\begin{proof}
We construct a relation $\rel{R} \subseteq S \times S$ where 
\[
\rel{R} = \{(p_0, q_0)\} \cup \{(q, p) \mid \exists g \in \playsConsFInit{f} : \exists k \in \mathbb{N} : \exists Q \subseteq S : g_k = \CDefSwapPos{p, Q} \land g_{k+1} = \attackerPos{q,\{p\}}\}.
\]
Informally, $\rel{R}$ contains the states $p_0, q_0$ of the initial position $\attackerPos{p_0, \set{q_0}}$ and the states of all attacker positions following a defender swap position in any play consistent with $f$. We aim to prove that $\rel{R}$ is a contrasimulation: 

Let $p, p', q \in S$ be states and assume $(p, q) \in \rel{R}$ and $p \wwvt p' $ for some $\vw \in Act^*$. 
We shall prove that a state $q' \in S$ exists such that $q \wwvt q'$ and $(q', p') \in \rel{R}$.

Since $(p,q) \in \rel{R}$, there exists a $g \in \playsConsFInit{f}$ and a $k \in \mathbb{N} \cup \{-1\}$ such that $g_{k+1}= \attackerPos{p,\{q\}}$ is an attacker position in $g$. We will distinguish between empty and nonempty words $\vw$:

\begin{itemize}
    \item Case 1: $\vw = \eps$.
    
    By assumption, we have $p \wepvt p'$ and thus $p \wtwt p'$. Hence, the attacker can move from $\attackerPos{p,\{q\}}$ to a defender position $\CDefSwapPos{p', \{q\}}$. Because $f$ is a winning strategy, the defender is not stuck at $\CDefSwapPos{p', \{q\}}$; there must therefore exist a state $q' \in S$ such that the defender can move to the position $f(g_0 \ldots \CDefSwapPos{p', \{q\}}) = \attackerPos{q',\{p'\}}$. 
    It follows that $q \wtwt q'$ and that there exists a play $g \in \playsConsFInit{f}$ in which $\attackerPos{q',\{p'\}}$ can be reached. Since the last position played is a $\CDefSwapPos{\ldots}$ node, we therefore have $(q', p') \in \rel{R}$ by our construction of $\rel{R}$.
    
    \item Case 2: $\vw = \vecv w_n$ for some $\vecv \in \wact$ and $w_n \in Act$. 
    
    By application of Lemma \ref{succs in play}, we know there exist $p_0, p_1$ with $p \wvvt p_0 \delay{w_n} p_1 \wtwt p'$ such that the position $\CDefSimPos{w_n, p_1,\, succs(\vecv, \{q\})}$ can be reached in some play $g \in \playsConsFInit{f}$. Because $f$ is a winning strategy, the defender is not stuck at $\CDefSimPos{w_n, p_1,\, succs(\vecv,\{q\}}$; there must therefore exist a set $Q_1$ allowing the defender to move to a position $f(g_0 \ldots \CDefSimPos{w_n, p_1,\, succs(\vecv,\{q\})}) = \attackerPos{p_1, Q_1}$. For $Q_1$ we have
    \[Q_1 = \{q_1 \mid succs(\vecv, \{q\}) \delay{w_{n}} q_1\} = succs(\vecv w_{n}, \{q\}) = succs(\vw, \{q\}).\]

    By our assumption $p_1 \wtwt p'$, the attacker can keep moving to a defender position $\CDefSwapPos{p', Q_1}$. Again, the defender is not stuck, and there exists a $q' \in S$ allowing the defender to move to the position $f(g_0 \ldots \CDefSwapPos{p', Q_1}) = \attackerPos{q', \{p'\}}$. It follows that $Q_1 \wtwt q'$.
    For $q'$ we have
    \begin{align*}
            q' &\in \{q' \mid \exists q_1 \in Q_1 : q_1 \wtwt q'\} \\
            &= \{q' \mid \exists q_1 \in succs(\vw, \{q\}) : q_1 \wtwt q'\}\\
            &\subseteq \{q' \mid \exists q_1 \in succs(\vw, \{q\}) : q \wwvt q_1 \land q_1 \wtwt q'\} && \text{(Lemma \ref{in succs iff word-reachable})}\\
            &\subseteq \{q' \mid q \wwvt q'\}.
    \end{align*}
    Hence, there exists a $q'$ satisfying $q \wwvt q'$. Because the defender moves according to $f$ from a position $\attackerPos{p,\{q\}}$ in a play $g \in \playsConsFInit{f}$, the position $\attackerPos{q', \set{p'}}$ must also be reachable in a play consistent with $f$. Then we also have $(q', p') \in \rel{R}$, since the last played position was a $\CDefSwapPos{\ldots}$ node. 
    
\end{itemize}
Thus, $\rel{R}$ is a contrasimulation by Definition \ref{contra def}. From $(p_0, q_0) \in \rel{R}$ then follows $p_0 \pre{C} q_0$.
\end{proof}

\subsection{Completeness of the Contrasimulation Game \label{subsec:contraim-game-complete}}

Let us now prove that the $\pre{C}$ game is complete with respect to the contrasimulation preorder, that is, the defender wins $\GM[\attackerPos{p, \set{q}}]$ if $p \pre{C} q$. To this end, we define an auxiliary function $F$ with which we are able to construct a defender strategy $\strat$ from the contrasimulation preorder $\pre{C}$. We will prove that if $p \pre{C} q$, then $\strat$ is a winning strategy for the defender on $\GM[\attackerPos{p, \set{q}}]$. 

We define the function $F : 2^{S\times 2^S} \rightarrow 2^{S\times 2^S}$ as follows:\isbref{definition}{Contrasimulation}{F}
\[F(R) = \{(p',\; succs(\vw, Q))\mid \exists p \in S: (p, Q) \in R \land p \wwvt p'\}
\]

\noindent
Furthermore, we define a type-congruent relation $C \subseteq S \times 2^S$ from the contrasimulation preorder with $C = \{(p, \{q\}) \mid p \pre{C} q\}$.

This yields a relation $F(C) \subseteq S \times 2^S$ which we will use extensively in the following proofs. 
The motivation behind $C$ and $F(C)$ becomes clear when we consider a play in the game $\GM[\attackerPos{p_0, \set{q_0}}]$ with $p_0 \pre{C} q_0$:
Here, $C$ contains tuples $(p, \set{q})$ of attacker positions after the word challenge has been completed, i.e. positions $\attackerPos{p, \set{q}}$ following a $\CDefSwapPos{\ldots}$ node. The relation $F(C)$ then expands $C$ to also include attacker positions in the \emph{simulation phase} of the word challenge, i.e. positions $\attackerPos{p, Q}$ following a $\CDefSimPos{\ldots}$ node. 
Thus, using $F(C)$ we can construct a defender strategy that is well-defined for both $\CDefSimPos{\ldots}$ and $\CDefSwapPos{\ldots}$ positions.

%Before defining such a strategy, let us first prove that $C$ is indeed contained in $F(C)$, and that 
%
%
\begin{lemma}
We have $R \subseteq F(R)$ for all $R \subseteq S \times 2^S$.\isbref{lemma}{Contrasimulation}{R_is_in_F_of_R}
\end{lemma}

%------------------set_type version-----------------

%Furthermore, we define a function $\mathit{set\_type} : 2^{S \times S} \xrightarrow[]{} 2^{S \times 2^S}$ that projects any relation $R \subseteq S \times S$ into the function space accepted by $\mathit{F}$: \[\mathit{set\_type}(R) = \{(p, \{q\}) \mid p, q \in R\}\]

%----------------/set_type version-----------------

\begin{lemma}
\label{F of C action or tau succ}

Let $a \in Act$ and let $p, p' \in S$ be states and $Q \subseteq S$ be a set of states such that $(p, Q) \in F(C)$.
\begin{itemize}
    \item If $p \delay{a} p'$, then there exists a set $Q' \subseteq S$ such that $Q \delay{a} Q'$ and $(p', Q') \in F(C)$,\isbref{lemma}{Contrasimulation}{F_of_C_guarantees_action_succ} and
    \item if $p \wtwt p'$, then there exists a state $q' \in S$ such that $Q \wtwt q'$ and $(q', \{p'\}) \in F(C).$\isbref{lemma}{Contrasimulation}{F_of_C_guarantees_tau_succ}
\end{itemize}
\end{lemma}
\begin{proof}

By construction of $F$ and our assumption $(p, Q) \in F(C)$, we know there exist states $p_0, q_0 \in S$ such that $(p_0, \{q_0\}) \in C$, $p_0 \wwvt p$ and $Q = succs(\vw, \{q_0\})$ for some $\vw \in \wact$. Hence, we also have $p_0 \pre{C} q_0$. We will distinguish between $\tau$-steps and delay steps: 

\begin{itemize}

        \item For $p \delay{a} p'$, we have $p_0 \wwvt p \delay{a} p'$ and thus $p_0 \wtrans{\vw a} p'$. With $p_0 \pre{C} q_0$, there must then exist a $q' \in S$ such that $q_0 \wtrans{\vw a} q'$. It follows from Lemma~\ref{in succs iff word-reachable} that $succs(\vw a, \{q_0\}) \wtwt q'$. For $Q' := succs(\vw a, \{q_0\})$, we then have $(p', Q') \in F(C)$ by our construction of $F$ and $Q \delay{a} Q'$ by Definition~\ref{succ func}.
        
        \item For $p \wtwt p'$, we have $p_0 \wwvt p \wtwt p'$ and thus $p_0 \wwvt p'$. With $p_0 \pre{C} q_0$, there must then exist a $q' \in S$ such that $q_0 \wwvt q'$ and $q' \pre{C} p'$. Thus, we have $(q', \{p'\}) \in C \subseteq F(C)$ and $q' \in succs(\vw, \{q_0\}) = Q$ by application of Lemma~\ref{in succs iff word-reachable}.  It follows immediately that $Q \wtwt q'$. \qedhere
        
\end{itemize}
%\vspace{-\topsep}
\end{proof}
\noindent
We can now construct a positional defender strategy $\strat$ using $F(C)$. 

\begin{defi}[Defender strategy $\strat$ on $\GM$]

Wherever possible, a defender strategy $\strat$ derived from $F(C)$ maps the current play fragment $g_0 \ldots g_k$ to next positions as follows:
%Let $g_0 \ldots g_k$ be the current play fragment on $\GM[g_0]$. We define a set of strategies where the defender's next move is dictated according to the following rules: 
%We define a defender strategy $\strat$, where $\strat(g_0 \ldots g_k)$ is given by the following rules: 
\begin{itemize}
    \item If the last played position $g_k$ is a swap node $\CDefSwapPos{p', Q}$, move to some attacker position $\attackerPos{q', \set{p'}}$ with $(q', \set{p'}) \in F(C)$ and $Q \wtwt q'$, and
    \item if the last played position $g_k$ is a simulation node $\CDefSimPos{a, p', Q}$, move to the attacker position $\attackerPos{p', Q'}$ where $Q \delay{a} Q'$.\isbref{fun}{Contrasim_Set_Game}{strategy_from_F_of_C}
    %\item if the last played position $g_k$ is a simulation node $\CDefSimPos{a, p', Q}$, move to some attacker position $\attackerPos{p', Q'}$ with $(p', Q') \in F(C)$ and $Q \delay{a} Q'$.\isbref{fun}{Contrasim_Set_Game}{strategy_from_F_of_C}
\end{itemize}
\end{defi}

\noindent
Where there are no applicable moves, we leave $\strat$ undefined. Note that there may exist several strategies satisfying these conditions. In the following, $\strat$ is one of these defender strategies---it makes no difference which one.\footnote{In our Isabelle/HOL formalization, we used its Hilbert's choice operator \texttt{SOME}.} 

\begin{comment}
\begin{equation*}
\strat(g_0...g_k)= 
    \begin{cases}
        \attackerPos{q', \set{p'}} & \parbox[t]{.45\columnwidth}{with ${q' = \epsilon q' .\, (q', \{p'\}) \in F(C) \land Q \wtwt q'}$ if $g_k = \CDefSwapPos{p', Q}$,}\\
        \attackerPos{p', Q'} & \parbox[t]{.4\columnwidth}{ with ${q' = \epsilon q' .\, (p', Q') \in F(C) \land {Q \delay{a} Q'}}$
        if $g_k = \CDefSimPos{a, p', Q}$,}\\
        \bot       & \text{else. }
    \end{cases}
\end{equation*}

\end{comment}

\begin{lemma}[$F$ invariant]
\label{f retains zwischenrel}
Let $p_0, q_0 \in S$ be states and let $g$ be a play of $\GM[\attackerPos{p_0, \set{q_0}}]$ consistent with $\strat$. If $p_0 \pre{C} q_0$, then we have $(p, Q) \in F(C)$ for all attacker positions ${\attackerPos{p, Q}}$ in $g$.\isbref{lemma}{Contrasim_Set_Game}{set_game_strategy_retains_F}
\end{lemma}
\begin{proof}
This follows immediately from $C \subseteq F(C)$ for the initial position $\attackerPos{p_0, \set{q_0}}$. All other attacker positions can only be reached if the defender moves to them, and, following $\strat$, the defender always moves in accordance with $F(C)$.
\end{proof} 

%\newpage

\begin{lemma}[Completeness]
\label{set game complete}
Let $p_0, q_0 \in S$ be states with $p_0 \pre{C} q_0$. Then $\strat$ is a winning strategy on $\GM[\attackerPos{p_0, \set{q_0}}]$.\isbref{lemma}{Contrasim_Set_Game}{set_contrasim_game_complete}
\end{lemma}
\begin{proof} 
We show by induction on $g \in \playsConsFInit{\strat}$ that the defender is never stuck. Let $g = g_0g_1 \ldots g_k$ be the current play fragment on $\GM[\attackerPos{p_0, \set{q_0}}]$.

At the initial position $g_0 = \attackerPos{p_0, \set{q_0}}$, the defender cannot be stuck, since $g_0$ is an attacker position. We will assume the lemma holds for the current play fragment $g_0...g_k$. For $g_{k+1}$ we have the cases:

\begin{itemize}
    \item $g_{k+1}$ is an attacker position. Then the defender cannot be stuck by the same reasoning as for the base case. 

    \item $g_{k+1}$ is a defender position.
    Then there exists a position $g_k = {\attackerPos{p, Q}}$ from which the attacker has moved to $g_{k+1}$, therefore ($p, Q) \in F(C)$ must hold by Lemma~\ref{f retains zwischenrel}.
    We will distinguish between types of defender positions for $g_{k+1}$: 
    \begin{itemize}
        \item $g_{k+1} = \CDefSimPos{a, p', Q}$ is a simulation position. Then we have $p \delay{a} p'$. It follows from ${(p, Q) \in F(C)}$ and Lemma~\ref{F of C action or tau succ} that there is a set $Q' \subseteq S$ such that $Q \delay{a} Q'$ and $(p', Q') \in F(C)$. Thus, the defender can move to $\attackerPos{p', Q'}$ with $\strat$ and is therefore not stuck.
        
        \item $g_{k+1} = \CDefSwapPos{p', Q}$ is a swapping position. Then we have $p \wtwt p'$. From $(p, Q) \in F(C)$ and Lemma~\ref{F of C action or tau succ}, it follows that there exists a state $q' \in S$ such that $Q \wtwt q'$ and $(q',\{p'\}) \in F(C)$. Thus, the defender can move to $\attackerPos{q', \{p'\}}$ with $\strat$ and is therefore not stuck.
  \end{itemize}
\end{itemize}
Hence, the defender is never stuck in a play consistent with $\strat$, and thus $\strat$ is a winning strategy for the defender.
\end{proof}

\noindent
Combining Lemma~\ref{soundness_menge} and Lemma~\ref{set game complete}, we get:

\begin{theorem}
The defender wins $\GM[\attackerPos{p, \set{q}}]$ precisely if $p \pre{C} q$.\isbref{theorem}{Contrasim_Set_Game}{winning_strategy_in_set_game_iff_contrasim}
\end{theorem}

\section{Discussion and Related Work \label{sec:related}}

The presented game closes the last interesting gap in the \emph{landscape of game characterizations for equivalences in the linear-time--branching-time spectrum with internal steps}~\cite{van1993linear}. De Frutos Escrig et al.'s games for branching, delay, $\eta$ and weak bisimulation~\cite{escrig_games_2017} and ours for coupled simulation~\cite{bisping2020coupled32} are polynomial in the size of the system state space. Due to the subset construction, the presented contrasimulation game needs exponentially many game positions. Compared to defining the equivalence games on words as Chen and Deng~\cite{ramalingam_game_2008}, the subset construction is still preferable, as it yields finite games for finite-state systems. Our group has used the same subset approach in~\cite{bn2021ltbtspectroscope}. There, the virtue of the single-action construction lay in making attacker strategies correspond closely with all relevant distinguishing Hennessy--Milner logic formulas.

A side-effect of our work is that we have formalized contrasimilarity (and its characterizing games) in \emph{Isabelle/HOL}, based on our prior work in~\cite{bisping2019isabelle}. This is not the first such formalization. Peters and van Glabbeek~\cite{pg15encodability,pg15afp} provided an Isabelle theory of contrasimilarity tailored to reduction semantics and to the analysis of encodings between formalisms. Moreover, Bell~\cite{bell_certifiably_2013} developed a Coq formalization to support the verification of compilers using contrasimilarity, still available at \url{http://people.csail.mit.edu/cj/par/}.

Voorhoeve and Mauw~\cite{voorhoeve_impossible_2001} examined a \emph{modal-logical characterization} of contrasimilarity consisting of $\top$, $\bot$, conjunction $\varphi \land \psi$, disjunction $\varphi \lor \psi$ and a special necessity operator on paths with immediate negation $\Box_W \neg \varphi$, which is true at $p$ iff $\varphi$ is true at all $p'$ with $p \wvt{w} p'$ and $\vec{w} \in W$.
Although this paper has not been about modal logic, there is a strong link between games and modal logic. One can read our game in Figure~\ref{fig:set_game} as a way of enumerating possible distinguishing Hennessy--Milner logic (HML) formulas (with $\hmlObs{\varepsilon}$ denoting points of possible internal behavior). Then, the $\texttt{Sim}$-branch corresponds to a delayed observation $\hmlObs{\varepsilon}\hmlObs{a}\varphi$ and the $\texttt{Swap}$-branch to a delayed logical nor $\hmlObs{\varepsilon}\neg(\varphi_1 \lor \varphi_2 \lor \ldots)$.\footnote{For more on why focusing the defender on one position is conjunction and swapping sides is negation (which together amounts to NOR), see our group's recent paper on linear-time--branching-time spectroscopy \cite{bn2021ltbtspectroscope}.} For instance, a distinguishing formula for the non-contrasimilar systems of Example~\ref{exa:philosophers-dining-late} would be $\hmlObs{\varepsilon}\neg\hmlObs{\varepsilon}\hmlObs{\action{op}}\hmlObs{\varepsilon}\hmlObs{\action{aEats}}$.

Observation $\hmlObs{a}\varphi$ and logical nor $\neg(\varphi_1 \lor \varphi_2 \lor \ldots)$ form a functionally complete set of operators for HML and thus characterize (strong) bisimilarity. So it is aesthetically pleasing that “weakening” this observation language by alternating its constructs and $\hmlObs{\varepsilon}$, thereby allowing internal behavior “in between the connectives” leads precisely to contrasimilarity. This view provides further evidence why contrasimilarity is a quite sensible way of generalizing bisimilarity to systems with internal steps as discussed in Section~\ref{sec:contrasim}.

Contrasimilarity is an only slightly coarser sibling of coupled similarity \cite{parrow1994complete}. For coupled similarity, our group has been able to prove that it can be decided in cubic time, and cannot be cheaper than deciding weak similarity \cite{bn2019coupledsimTacas}. The exponential game of the present paper induces an exponential algorithm for deciding contrasimilarity. We have integrated the contrasimilarity algorithm into our prototypical coupled similarity checker on \url{https://coupledsim.bbisping.de}.
So, a side effect of the work presented here is to provide the first tool support for checking contrasimilarity.

It seems strange that deciding contrasimilarity, which coincides with coupled similarity in many cases, should be so much more expensive. We have not yet thought of a decisive argument for the complexity of contrasimilarity checking. Due to the results for the spectrum without internal steps~\cite{hs1996complexityEquivalences}, one may presume the coarser siblings of contrasimilarity from weak impossible futures down to weak trace equivalence to all be PSPACE-hard. But while deciding weak possible futures and nestings of possible futures equivalence is PSPACE, deciding its arbitrarily nested limit, weak bisimilarity, is PTIME. So how come that the arbitrarily nested limit of impossible futures equivalence, contrasimilarity, seems to be beyond PTIME?

The possibility of fooling depth-restricted approximations of contrasimilarity using instable external choice points as in Example~\ref{exa:instable-choice} suggests that one indeed must cater for the complexity of handling words. We would very much like to be proven wrong with regard to this. Also, not every formalism is expressive enough to raise this issue. In order to avoid the problem, it might suffice to ensure that all observations are committed, which would rule out Example~\ref{exa:instable-choice}. Fournet and Gonthier~\cite{fg2005hierarchy} showed that parts of the hierarchy of equivalences around contrasimilarity collapse for reduction semantics if transient barbs are excluded. Then, however, cheaper algorithms for easier predicates such as coupled simulation should also be applicable right away. More research would be needed in order to pinpoint where exactly one cannot do without arbitrary-length word contrasimulation.

For this, and all the other nice things one could do using contrasimilarity, everyone is welcome to use our Isabelle/HOL formalization.\footnote{Available from \url{https://github.com/luisamontanari/ContrasimGame}.}

{\small\medskip\noindent{\bf Acknowledgments.} We are thankful to the members of the TU Berlin research group Modelle und Theorie Verteilter Systeme, especially Uwe Nestmann, for supporting us in the preparation of this paper. Many thanks also to the EXPRESS/SOS reviewers!}

\bibliographystyle{eptcs}
\bibliography{sources}

\end{document}